\documentclass[a4paper,12pt]{article}

\usepackage{amsmath}
\usepackage{amsbsy}
\usepackage{amsthm}
\usepackage{amsfonts}

\usepackage{tikz}
\usetikzlibrary{decorations.markings}

\newtheorem{theorem}{Theorem}
\newtheorem{lemma}[theorem]{Lemma}

\DeclareMathOperator{\Log}{Log}

\newcommand{\T}{T}

 \newcommand{\tkNode}[1]{
\protect \tikz[baseline=-0.4em]{ \protect \node [#1] {};}
}

\title{The Ising model on a cylinder: universal finite size corrections and diagonalized action}
\author{Rafael L.\ Greenblatt}
\date{\today}

\begin{document}
\maketitle

\begin{abstract}
Finite size corrections to the pressure (free energy) of the Ising model on a 2 dimensional cylinder are calculated and shown to be consistent with the predictions of conformal field theory.  The exact solution of the model is expressed in terms of the determinant of a block-diagonal matrix.
\end{abstract}

Among the many implications of the widely accepted hypothesis that the critical behavior of two-dimensional models in statistical mechanics is described by conformal field theory is the presence of a universal finite-size correction term in the pressure (logarithm of the partition function) at any critical point \cite{Affleck.universal,Cardy.universal}.  For an $M \times N$ system with $N >> M >> 1$, the pressure is expected to have the form
\begin{equation}
\log Z \approx p MN + k c \frac{N}{M},  \label{eq:cardy}
\end{equation}
where $c$ is the central charge of the appropriate CFT and $k$ is a universal factor depending only on the boundary conditions in the $M$ direction, taking the value $\pi/6$ for periodic boundary conditions (an infinite cylinder) and $\pi/24$ for open boundary conditions (an infinite strip).  In the Ising model, where it is expected that $c=\frac12$ based on the scaling forms of the correlation functions \cite{BPZ}, it was noted already in~\cite{Cardy.universal} that in the case of fully periodic boundary conditions Equation~\eqref{eq:cardy} can be obtained explicitly from the exact solution of the model, using a calculation due to Ferdinand and Fisher \cite{FF}.  Subsequently, Lu and Wu~\cite{LuWu} made a similar calculation on a M{\"o}bius strip and a Klein bottle, and Izmailian et.\ al.\ \cite{Izmailian.Brascamp} did likewise for the more exotic Brascamp-Kunz boundary conditions.   
Nonetheless, no results of this sort appear so far in the literature for the Ising model on an infinite cylinder, in spite of the fact that the exact solution in the case of cylindrical boundary was already studied in McCoy and Wu's classic monograph \cite{MW}.  This solution expresses the partition function as the Pfaffian (that is, the square root of the determinant) of a matrix which we will call the action, by analogy with the path integral representation of a free Fermionic field (the analogy extends to the form of many correlation functions \cite{Us}). Depending on the way in which the system size is taken to infinity, these boundary conditions give either an infinite strip or an infinite cylinder. 

Recently, Giuliani and Mastropietro \cite{GM13} used constructive renormalization group techniques to show that the expansion in Equation~\eqref{eq:cardy} is also valid for a large variety of non-solvable variants of the Ising model, however their proof holds only for fully periodic boundary conditions.  Among the reasons for this limitation is the fact that their technique, which maps additional terms in the Ising Hamiltonian into terms analogous to interactions in a Fermionic field, 
 requires an explicit diagonalization (or block-diagonalization in blocks of small, fixed size, which is much the same thing) of the action in the solution of the Ising model.  In the periodic case, this is a block-Toeplitz matrix and therefore can be block-diagonalized by a Fourier transform, but there is no satisfactory technique in general.

In Section~\ref{sec:solution}, I will review the exact solution of the Ising model in cylindrical boundary conditions.  This section concludes by expressing the action terms of a certain tridiagonal matrix.  
Section~\ref{sec:diag} shows how this matrix can be transformed into a suitable block-diagonal form, which is given explicitly up to the determination of the roots of a certain polynomial given in Equation~\eqref{eq:eigen_poly}.  Although the roots of this polynomial do not themselves have an explicit formula, their properties (which have already been noted in another context \cite{XXZ.open}) are spelled out in great detail in Lemma~\ref{lem:roots}.  Some explicit calculations can be made using this form: in Section~\ref{sec:diag_final}, I use it to rederive the exact formula for the partition function.  Finally, in Section~\ref{sec:uni}, I provide an expansion of the logarithm of the partition function which verifies Equation~\eqref{eq:cardy} for both the infinite strip and infinite cylinder geometries.  

As in the exact calculations for other boundary conditions  \cite{FF,LuWu,Izmailian.Brascamp}, Equation~\eqref{eq:cardy} appears as a limiting case of the expansion
\begin{equation}
\log Z = p MN + \kappa(M/N) + O(1/N) \label{eq:kappa_exp}
\end{equation}
where $M$ and $N$ are comparable in size, and $\kappa$ is given explicitly in terms of Jacobi theta functions in Equation~\eqref{eq:kappa_final}.  Unsurprisingly, the function $\kappa$ obtained here for cylindrical boundary conditions is different from those obtained for other boundary conditions.


\section{The exact solution of the Ising model on a cylinder}\label{sec:solution}

In this section I will review the exact solution of the Ising model in cylindrical boundary conditions due to McCoy and Wu \cite[p.~113-20]{MW}.  If we restrict to the case of isotropic interactions, the system is defined by the Hamiltonian
\begin{equation}
H(\sigma) = -J \sum_{(x,y) \sim (x',y')} \sigma_{xy} \sigma_{x'y'}
\end{equation}
where the sum runs over pairs of sites $(x,y), (x',y') \in \mathbb{Z}_M \times \mathbb{Z}_N$ which are nearest neighbors, including periodic boundary conditions in the $y$ direction.  We denote $t = \tanh \beta J$; with this notation the critical point of the system is $t = t_c := \sqrt{2} -1$.  Note that McCoy and Wu also allowed an additional term in the Hamiltonian coupling to the sites on one of the boundaries, but we take this term to be zero.

For $N$ even the partition function of this system can be expressed as
\begin{equation}
Z = \frac12 \left(2 \cosh \beta J\right)^{MN} \left( \cosh \beta J \right)^{N(M-1)} \sqrt{|S|}
\label{eq:Z_primitive} \end{equation}
where $S$, which we call the action, is an antisymmetric $4MN \times 4MN$ matrix $S$ made up of $4 \times 4$ blocks 
\begin{gather}
S_{x,y;x,y} = \begin{bmatrix}
0 & 1 & -1 & -1 \\
-1 & 0 & 1 & -1 \\
1 & -1 & 0 & 1 \\
1 & 1 & -1 & 0 
\end{bmatrix}, \; 1 \le x \le M, \; 1 \le y \le N \\
S_{x,y;x,y+1} = - S^\T_{x,y+1;x,y} = \begin{bmatrix}
0 & t & 0 & 0 \\
0 & 0 & 0 & 0 \\
0 & 0 & 0 & 0 \\
0 & 0 & 0 & 0 
\end{bmatrix}, \; 1 \le x \le M, \; 1 \le y < N \\
S_{x,y;x+1,y} = - S^\T_{x+1,y;x,y} = \begin{bmatrix}
0 & 0 & 0 & 0 \\
0 & 0 & 0 & 0 \\
0 & 0 & 0 & t \\
0 & 0 & 0 & 0 
\end{bmatrix}, \; 1 \le x < M, \; 1 \le y \le N \\
S_{x,N;x,1} = - S^\T_{x,1;x,N} = \begin{bmatrix}
0 & - t & 0 & 0 \\
0 & 0 & 0 & 0 \\
0 & 0 & 0 & 0 \\
0 & 0 & 0 & 0 
\end{bmatrix}, \; 1 \le x \le M \label{eq:matrix_boundary}
\end{gather}
and all other entries zero.  It is easy to confirm that the generalization to $N$ odd involves only changing the sign of the terms defined in Equation~\eqref{eq:matrix_boundary}, but I will consider only the even case in order to avoid complicating my notation.

Carrying out a Fourier transform in the $y$ direction block-diagonalizes $S$ in $N$ $4M \times 4M$ blocks, and a further transformation decomposes each of these into two blocks, giving 
\begin{equation}
|S| =  \prod_k |1 + t e^{ik}|^{2M} |A_M(k)|
\label{eq:s_det}
\end{equation}
where the sum over $k$ runs over
\begin{equation}
k = \frac{\pi(2n-1)}{N}
\end{equation}
for $n = 1,\dots,N$, and
$A_M(k)$ is the $2M \times 2M$ matrix \cite[p.~120]{MW}
\begin{equation}
A_M (k) :=\begin{bmatrix}
a(k) & b(k) & \\
-b(k) & - a(k) &  t \\
& -t & a(k) & b(k) \\
& &-b(k) & - a(k) &  t \\
& & & -t &a(k) & \ddots \\
& & & & \ddots & \ddots
\end{bmatrix}, \label{eq:Tform}
\end{equation}
where
\begin{gather}
a(k) := -\frac{2 t i \sin k}{|1+te^{ik}|^2} \label{eq:a_def},\\
b(k) := \frac{1-t^2 }{|1+te^{ik}|^2} . \label{eq:b_def}
\end{gather}

\section{Diagonalization of the matrix $A_M$} \label{sec:diag}
In~\cite{MW}, formulae are given for the determinant and inverse of $A_M$, but the matrix is not diagonalized.  As shown in this section, it can be block-diagonalized by a transformation which can be thought of as a Fourier sine transformation with modified frequencies.

In this section I will write $A_M(k) = A_M, a = a(k)$, $b = b(k)$ for brevity, since dependence on $k$ plays no role.

It is helpful to begin by diagonalizing the real symmetric matrix
\begin{equation}
A_M^2 = \begin{bmatrix}
a^2-b^2 & 0 & bt \\
 0 & a^2-b^2-t^2 & 0 & \ddots \\
bt & 0 & a^2-b^2-t^2 & \ddots & bt \\
&\ddots & \ddots & \ddots & 0 & bt\\
&& bt & 0 & a^2-b^2 -t^2 & 0\\
&&& bt & 0 & a^2-b^2
\end{bmatrix}
\end{equation}
with all diagonal entries other than the first and the last equal.  Note that all matrix entries between even rows and odd columns (and vice versa) vanish, so this matrix becomes block diagonal after a suitable rearrangement, with the blocks given by
\begin{equation}
B_M:=\begin{bmatrix}
a^2-b^2 & bt \\
bt & a^2 -b^2 -t^2 & bt \\
& bt & a^2-b^2-t^2 & \ddots \\
&& \ddots & \ddots
\end{bmatrix}
\end{equation}
and the matrix $\tilde{B}_M$ given by reversing the order of the rows and columns of $B_M$.  $B_M$ is very similar to a discrete Laplacian with peculiar boundary conditions, which suggests the ansatz
\begin{equation}
v_z=\begin{bmatrix}
\alpha_z z + \beta_z z^{-1} \\
\alpha_z z^2 + \beta_z z^{-2} \\
\vdots
\end{bmatrix},\label{eq:vdef}
\end{equation}
which is an eigenvector of $B_M$ iff the system of equations
\begin{gather}
(a^2-b^2)(\alpha_z z + \beta_z z^{-1}) + bt (\alpha_z z^2+ \beta_z z^{-2}) = \lambda_z (\alpha_z z + \beta_z z^{-1})\\[2ex]
\begin{split}
 bt (\alpha_z z^{k-1} +&\beta_z z^{-k+1})+ (a^2-b^2-t^2)(\alpha_z z^k + \beta_z z^{-k})  \\
  + bt &(\alpha_z z^{k+1} + \beta_z z^{-k-1})  \\
   &= \lambda_z (\alpha_z z^k + \beta_z z^{-k}), \qquad  1<k<M \label{eq:bulk_eigen}
\end{split}\\[2ex]
\begin{split}
 bt (\alpha_z z^{M-1} +\beta_z z^{-M+1})+ (a^2-b^2-t^2)(\alpha_z z^M + \beta_z z^{-M})  \\ = \lambda_z (\alpha_z z^M + \beta_z z^{-M})
\end{split}
\end{gather}
are all satisfied. Equation \eqref{eq:bulk_eigen} is solved by choosing $$\lambda_z = bt(z+z^{-1}) + (a^2-b^2-t^2),$$ which reduces the other two conditions to 
\begin{gather}
bt (\alpha_z +\beta_z) - t^2 (\alpha_z z + \beta_z z^{-1}) = 0  \label{eq:bd_eig}\\
bt  (\alpha_z z^{M+1} + \beta_z z^{-M-1}) = 0 
\end{gather}
  
The last condition implies
\begin{equation}
\beta_z = -\alpha_z z^{2M+2}\label{eq:ab_condition}
\end{equation}
which can be used to rewrite Equation~\eqref{eq:bd_eig} as
\begin{equation}
z^{2M+2} - \frac{t}{b}  z^{2M+1} + \frac{t}{b} z - 1 = 0. \label{eq:eigen_poly}
\end{equation}
The roots of this polynomial cannot generally be expressed explicitly, but they can be described in great detail.  I will postpone this discussion to Lemma~\ref{lem:roots} at the end of this section so as not to interrupt the flow of this calculation, and note a few properties which are immediately relevant.  Apart from $\pm 1$, where $v_{\pm 1} = 0$, the roots of Equation~\eqref{eq:eigen_poly} come in pairs $\{z, z^{-1}\}$ with either $|z|=1$ or $z$ real.  Each such pair corresponds to a different eigenvalue, and therefore to a linearly independent eigenvector $v_z$ of $B_M$.  When all the roots are nondegenerate, which is the case apart from a single value of $M$ (depending on $t/b$), this gives a complete set of $M$ eigenvectors.  The degenerate case can be avoided by skipping certain values of $M$ and $N$ in the thermodynamic limit, or by appealing to piecewise continuity of the quantities being calculated.

The resulting eigenvectors of $A_M^2$ are
\begin{equation}
u_z = c_z 
\begin{bmatrix}
z^{-M} - z^M \\ 0 \\
z^{1-M} - z^{M-1} \\ 0 \\ \vdots \\
z^{-1} - z \\ 0
\end{bmatrix}, \ 
w_z = c_z 
\begin{bmatrix}
0 \\ z^{-1} - z \\
0 \\ z^{-2} - z^2 \\ \vdots \\
0 \\ z^{-M} - z^M
\end{bmatrix}
\end{equation} 
where $z$ runs over a set $R_M$ of $M$ distinct roots of Equation~\eqref{eq:eigen_poly} such that each eigenvalue appears once, and $|c_z|=\alpha_z z^{M+1}$ is a normalization factor. 
To make $u_z$ and $w_z$ real, it suffices to take $c_z$ real for $z$ real and pure imaginary for $z$ on the unit circle.

We now return to $A_M$.
Noting that $$bz^{M+1} - t z^M = bz^{-M-1} - tz^{-M}$$ (from Eq.~\eqref{eq:eigen_poly}),
\begin{equation}
\begin{split}
A_M u_z = a u_z 
+ c_z \begin{bmatrix}
0 \\ -b (z^{-M} - z^M) + t (z^{1-M} - z^{M-1}) \\
0 \\ -b (z^{1-M} - z^{M-1}) + t (z^{2-M} - z^{M-2}) \\ \vdots \\
0 \\ -b (z^{-1} - z) + t(1-1)
\end{bmatrix}
\\ = a u_z 
+ ( b z^{n+1} -t z^{n} ) w_z,
\end{split}
\end{equation}
and
\begin{equation}
\begin{split}
A_M w_z = -a w_z + c_z \begin{bmatrix}
-t (1 - 1) + b (z^{-1} - z) \\ 0 \\
-t (z^{-1} - z) + b (z^{-2} - z^2) \\ 0 \\ \vdots \\
-t (z^{1-M} -z^{M-1}) + b (z^{-M} - z^M) \\ 0
\end{bmatrix}
\\ = -a w_z + c_z \begin{bmatrix}
-t (z^{2M} - 1) + b (z^{2M+1} - z) \\ 0 \\
-t (z^{2M-1} - z) + b (z^{2M} - z^2) \\ 0 \\ \vdots \\
-t (z^{M+1} -z^{M-1}) + b (z^{M+2} - z^M) \\ 0
\end{bmatrix}
\\ = -aw_z + (tz^{M} - b z^{M+1})u_z
\end{split}
\end{equation}
or in other words, the change of variables given by $u_z$ and $w_z$ puts $A_M$ in block-diagonal form, with $2\times 2$ blocks
\begin{equation}
\begin{bmatrix}
a & z^{-M-1} (b - t z) \\
z^{M+1} (tz^{-1}-b) & -a
\end{bmatrix}.
\end{equation}

\newcommand{\PN}{\ensuremath{P_{M\beta}}}
\newcommand{\PNT}{\ensuremath{\tilde{P}_{M\beta}}}
\newcommand{\pn}{\ensuremath{p_M}}
\newcommand{\pnt}{\ensuremath{\tilde p_M}}
\newcommand{\rn}{\ensuremath{r_M}}
\newcommand{\rnt}{\ensuremath{\tilde r_M}}

%

Having arrived at this point, it is time to return to the question of the roots of Equation~\eqref{eq:eigen_poly}.  This polynomial appears as a factor of the one-particle Bethe equation for the Heisenberg XXZ chain in open boundary conditions, and its properties are not hard to establish \cite{XXZ.open}.
As the following Lemma shows, all or all but two of the roots are approximately evenly spaced around the unit circle. 

\begin{lemma} \label{lem:roots}
The polynomial
\begin{equation}
\PN(z) = z^{2M+2} - \beta z^{2M+1} + \beta z - 1 = 0,
\end{equation}
where $\beta > 0$, has $2M-2$ simple roots of the form $e^{\pm i q_j}$, with $$ \frac{\pi j }{ M} < q_j < \frac{\pi (j+1) }{ M+1}$$ for $j = 1,\dots,M-1$.  The other roots are either
\begin{enumerate}
\item $\pm 1$ and two complex roots of the form $e^{\pm iq_0}$ with $0 < q_0 < \pi/(M+1)$, if $\beta <  \frac{M+1}{M}$,
\item $\pm 1$ if $\beta = \frac{M+1}{M}$, \label{case:degen}
\item or $\pm 1$ and two positive real numbers $x$ and $1/x$, if $\beta >  \frac{M+1}{M}$.
\end{enumerate}
All roots are simple, except in case~\ref{case:degen}, when $1$ is the only degenerate root.
\end{lemma}
\begin{proof}

Let $\pn(z) =  z^{2M+2} - 1$ and $\rn(z) = z^{2M+1}-z$, so that
\begin{equation}
\PN (z) =  \pn(z) - \beta \rn(z).
\end{equation}
If we define 
\begin{equation}
\PNT(z) = i (z + i)^{2M+2} \PN \left( \frac{z-i}{z+i} \right) \label{eq:tilde_poly}
\end{equation}
and similarly $\pnt$, $\rnt$, then these are all polynomials, and $\PNT(z) = \pnt(z) + \rnt(z)$.  More precisely,
\begin{gather}
\pnt (z) = i (z-i)^{2M+2} - i (z+i)^{2M+2} = (4M+4) z^{2M+1} + O(z^{2M})\\
\rnt (z) = i (z-i)^{2M+1}(z+i) - i (z-i)(z+i)^{2M+1} = 4M z^{2M+1} + O(z^{2M})
\end{gather}
which are both odd polynomials with real coefficients, as is $\PN$.  $\pnt$ and $\rnt$ are both degree $2M+1$, and so $\PNT$ is at most of the same degree.  It is clear that the roots of $\PNT$ (resp.\ $\tilde p_n$, $\tilde r_n$) are the preimages under the M{\"o}bius transformation $z \mapsto \frac{z-i}{z+i}$ of the roots of $\PNT$ (resp $\rnt$, $\pnt$).

The roots of $\pn$ are the $e^{j \pi / (M+1)}$ ($j = -M,\dots,M+1$), and the roots of $\rn$ are $0$ and $e^{j \pi / M}$ ($j = -M+1,\dots,M$).  Noting that the M{\"o}bius transformation continuously maps the real line onto $\text{(the unit circle)} \setminus {1}$, we see that $\pnt$ has $2M+1$ real roots, which we denote by $-\pi_M < \dots < -\pi_1 < 0 < \pi_1 < \dots \pi_M$, while $\rnt$ has $2M-1$ real roots $-\rho_{M-1} < \dots < -\rho_1 < 0 < \rho_1 < \dots < \rho_{M-1}$.  Furthermore, $\pi_j < \rho_j < \pi_j+1$ for all $j <M$.  $\pm i$ are also roots of $\rnt$, so all of these roots must be nondegenerate.

Noting that $\pnt(x)$ and $\rnt(x)$ are both positive for large positive $x$, we see that $\PNT(x)$ is negative on $[\rho_{M-1}, \pi_M]$ (where $\pnt(x) \le 0$ and $\rnt(x) > 0$), positive on $[\rho_{M-2}, \pi_M-1]$, and so on.  For these signs to be obtained, $\PNT$ must have an odd number of roots in each of the intervals $(\rho_{j-1},\pi_{j})$, $j = 1,\dots,M-1$.  Bearing in mind that $\PNT$ is an antisymmetric polynomial of degree no more that $2M+1$, a counting argument shows that it cannot have more than one root in any of those intervals.  Bearing in mind that the M{\"o}bius transformation is continuous, we see that each root lies on an arc between the appropriate roots of $\pn$ and $\rn$; these and their complex conjugates are the $M-1$ pairs of roots of $\PN$ of the form $e^{\pm i q_j}$.

It is evident that $\pm 1$ are always roots of $\PN$, which can leave only two roots unaccounted for.  Examining the three cases enumerated above:
\begin{enumerate}
\item If $\beta < \frac{M+1}{M}$, then, for sufficiently large positive $x$, $\pnt(x) > \rnt(x)$ and $\PNT(x) > 0$.  Then $\PNT$ must have a real root in the interval $\pi_M,\infty$, whose preimage is of the desired form.
\item If $\beta = \frac{M+1}{M}$, it is easily verified that $\PN'(1) = \PN''(1) = 0$, so $1$ is a degenerate root of $\PN$ with multiplicity 3.
\item If $\beta > \frac{M+1}{M}$, it is easily verified that $\pm 1$ are simple roots of $\PN$.  $\PNT(x) < 0$ for sufficiently large positive $x$, so the only real roots of $\PNT$ are those already enumerated, and it must then have a single pair of imaginary roots.  Since $\PNT$ is odd, these must be pure imaginary, and the corresponding roots of $\PN$ are a pair of real numbers $x$ and $1/x$.

Without loss of generality we may take $|x| > 1$. 
$x$ cannot be negative, since,  for $x< -1$, $x^{2M+2} - 1 > 0$ and $x^{2M} -1 >0$ so $$\PN(x) = (x^{2M+2} - 1) - \beta x (x^{2M} -1) > 0.$$  
Similarly, for $x \ge \beta$, $x-\beta \ge 0$ and $\beta x -1 > 0$, so $$\PN(x) = (x-\beta)x^{2M+1} + (\beta x -1 ) > 0,$$ 
leaving only $1 < x < \beta$.
\end{enumerate}
\end{proof}

\section{Calculation of the partition function}\label{sec:diag_final}
Using the results of Section~\ref{sec:diag}, we can rewrite the expression~\eqref{eq:s_det} for the determinant appearing in the partition function in Equation~\eqref{eq:Z_primitive} as 
\begin{equation}
\begin{split}
|S| &= \prod_k |1+te^{ik}|^{2M} |A_M(k)| 
\\ &= \prod_k |1+te^{ik}|^{2M} \prod_{z \in R(M,k,t)}
\begin{vmatrix}
a(k) & z^{-M-1}(b(k)-t z) \\
z^{M+1}(t z^{-1} - b(k) & -a(k)
\end{vmatrix}.
\end{split}
\label{eq:s_final}
\end{equation}
where $R(M,k,t)$ is a subset of the roots of Equation~\eqref{eq:eigen_poly} in Section~\ref{sec:diag}, consisting of one of each pair $z, 1/z$ apart from $\pm 1$.
As shown in Lemma~\ref{lem:roots}, $z$ is not in all cases on the unit circle. A real $z$ ($\ne \pm 1$) will be present for $M$ sufficiently large exactly when $t/b(k) >1$, i.e.
\begin{equation}
\frac{t|1+te^{ik}|^2}{1-t^2} > 1.
\end{equation}
The left hand side of this expression is at its largest for $k=0$, where the inequality simplifies to
\begin{equation}
t^2 + 2t -1 > 0
\end{equation}
or rather (since $t > 0$)
\begin{equation}
t > \sqrt{2} -1 = t_c
\end{equation}
so such roots appear for $t > t_c$ (that is, in the ferromagnetic phase) for sufficiently small $k$. 

The fact that $z$ are not known explicitly would seem to limit the usefulness of this representation.  The situation is not as bad as it might seem, as we shall now see by calculating the partition function.  As before, we will omit the dependence of various quantities on $k$ when it is unimportant.

To evaluate the determinant of $A_M$, I write
\begin{equation}
\begin{split}
\log \det A_M = \sum_{z \in R(M,k,t)} \log \begin{vmatrix}
a & z^{-M-1} (b-tz) \\
z^{M+1} (t z^{-1} - b) & -a
\end{vmatrix} \\
= \sum_{z \in R(M,k,t)} \log \left(
-a^2+b^2+t^2 - bt [z+z^{-1}]
\right)
 =:  \sum_{z \in R(M,k,t)} \log f(z).
\end{split}
\end{equation}
We can relate the sum in this equation to a contour integral as follows.  Let $P_M = z^{2M+2} - \frac{t}{b} z^{2M+1} +  \frac{t}{b} z - 1$ be the polynomial in Equation~\eqref{eq:eigen_poly}.  Then for the parameter values where $P_M$ has no repeated roots, the function $z \mapsto P_M'(z) / P_M(z)$ is meromorphic, its poles are the zeros of $P_M$, and they are all simple poles with residue one, so that for any suitably analytic function $F(z)$ satisfying $F(z) = F(1/z)$ we have
\begin{equation}
\frac1{2 \pi i} \oint_C \frac {P'_M (z) F(z) }{P_M (z) } dz = F(1) + F(-1) + 2 \sum_{z \in R(M,k,t)} F(z)
\label{eq:sum_as_int}
\end{equation}
for any contour $C$ surrounding all the zeros of $P_M$.

In the case at hand we would like apply this expression to $F(z) = \log f(z)$;
this indeed satisfies $F(z) = F(1/z)$.  This function has branch points at $0$ and $\infty$ (where $f(0) = f(\infty) = \infty$) and at the zeros of $f$, which are 
\begin{equation}
z_\pm := \frac{
b^2 + t^2 - a^2 \pm \sqrt{(b^2 + t^2 - a^2)^2 - 4 b^2t^2}
}
{2bt}.
\label{eq:z_pm}
\end{equation}
From this expression it is clear that $z_+ > 1$ and $z_+ > t/b$.  If we choose to define $F(z)$ using the principal branch of the logarithm, the branch cuts of $F$ and the roots of $P_M$ are arranged as shown in Figure~\ref{fig:poles_branches}.
\begin{figure}
\centering
\tikzset{ zero/.style={draw,shape=circle, fill=black, inner sep = 1pt},
cross/.style={path picture={ 
  \draw
(path picture bounding box.south east) -- (path picture bounding box.north west) (path picture bounding box.south west) -- (path picture bounding box.north east);}},
branch_point/.style={draw,solid,shape=circle,cross,fill=white, inner sep = 3 pt} 
}

\begin{tikzpicture}[scale=1.7,radius=0.18,decoration={markings,mark=at position 0.5 with {\arrow [scale=2]{>}; } }]
\path[draw=black,dashed] (0,0) node [branch_point,label=$0$] {} -- (0.346412704171355,0) node [branch_point,label=$z_-$] (zm) {} ;
\draw (zm) +(0,0.18) arc[start angle = 90, delta angle = -180] decorate { -- (0,-0.18) }  arc[start angle = 270, delta angle = -180] (zm) +(0,0.18) -- (0,0.18) (zm) + (0,-0.18) -- (0,-0.18);
\path[draw=black,dashed] (2.88673015729049,0) node [branch_point,label=$z_+$] (zp) {} -- (5,0) node [label=right:\text{(to $\infty$) } ] {} ;
\draw (zp) + (-0.18,0) arc[start angle = 180, delta angle = -90] -- (4,0.18) decorate { arc[start angle = 0, delta angle = 15,radius = 4]};
\draw (4,0.18) arc[start angle = 0, delta angle = 15,radius = 4];
\draw (zp) + (-0.18,0) arc[start angle = 180, delta angle = 90] -- (4,-0.18) arc [start angle = 0, delta angle = -15, radius=4] decorate{arc [start angle = -15, delta angle = 15, radius = 4]};
\draw (-1.00000000000000,0.000000000000000) node [zero] {};
\draw (0.673828460681826,0.000000000000000) node [zero] {};
\draw (1.00000000000000,0.000000000000000) node [zero] {};
\draw (1.48405723169979,0.000000000000000) node [zero] {};
\draw (-0.835321942225124,-0.549761087052591) node [zero] {};
\draw (-0.835321942225124,0.549761087052591) node [zero] {};
\draw (-0.393604905748351,-0.919279706167188) node [zero] {};
\draw (-0.393604905748351,0.919279706167188) node [zero] {};
\draw (0.184747619302029,-0.982785997642535) node [zero] {};
\draw (0.184747619302029,0.982785997642535) node [zero] {};
\draw (0.715236382480640,-0.698882620456403) node [zero] {};
\draw (0.715236382480640,0.698882620456403) node [zero] {};

\end{tikzpicture}
\caption{Contour of integration and points of non-analyticity of the integrand in Equation~\eqref{eq:sum_as_int}.  Zeros of $P_M$ are indicated by \tkNode{zero}, branch points of $F(z)$ by {\tkNode{branch_point}} and branch cuts by dashed lines.}  \label{fig:poles_branches}
\end{figure}
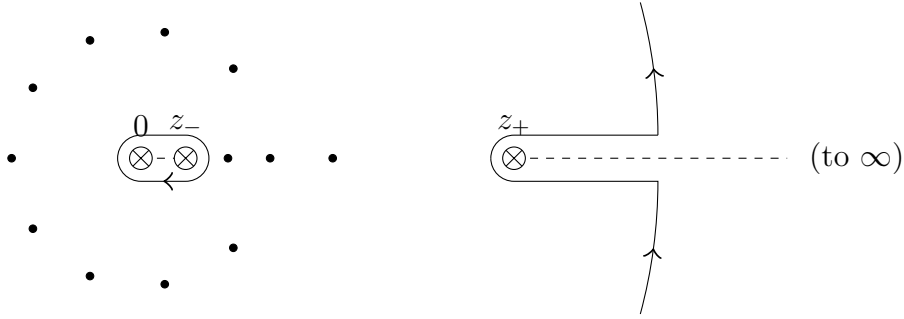
$C$ will be the (disconnected) contour also shown there.  The first component of $C$ consists of two parts: a circle of of radius $R > z_+$ interrupted at the branch cut of $F(z)$, which I denote by $C_R$, and a trajectory $B_R$ which moves leftwards from $R-i0$ around $z_+$ to $R+i0$.  Since the branch cut is that of the logarithm, $F(z)$ changes by $2\pi i$ across the branch cut, and thus
\begin{equation}
\int_{B_R} \frac {P'_M (z) F(z) }{P_M (z) } dz = - 2 \pi i \int_{z_+}^R  \frac {P'_M (z) }{P_M (z) } dz = -2 \pi i \left( \log \frac{ P_M(R)}{P_M(z_+)} \right).
\end{equation}

To evaluate the integral along $C_R$, we note that 
\begin{equation}
\frac{P'_M (z) }{P_M (z) } = \frac{2M+2}{z} + O\left( \frac1{z^2} \right)
\end{equation}
and that
\begin{equation}
F(z) = \Log z + \log bt +  O\left( \frac1{z} \right),
\end{equation}
and therefore
\begin{equation}
\begin{split}
\int_{C_R} \frac {P'_M (z) F(z) }{P_M (z) } dz = (2M+2) \int_{C_R} \left[ 
\frac{\Log z}{z} +\frac{\log bt}{z} +O\left( \frac1{z^2}\right)
\right]
\\=(2M+2) \int_{C_R} \left[ 
\frac{\log R}{z} + i \frac{\arg z}{z}+\frac{\log bt}{z} +O\left( \frac1{z^2}\right)
\right]. 
\end{split} \label{eq:CR_expansion}
\end{equation}
Noting that
\begin{equation}
\int_{C_R} \frac{\arg z}{z} dz = \frac1{R} \int_0^{2\pi} \theta e^{i \theta} d \theta = \frac{2 \pi i}{R}
\end{equation}
and
\begin{equation}
\log P_M(R) = (2M+2) \log R + O\left( \frac1R \right),
\end{equation}
we can combine the preceding equations to obtain
\begin{equation}
\frac1{2 \pi i} \oint_{B_R+C_R} \frac {P'_M (z) F(z) }{P_M (z) } dz = (2M+2)  \log bt  + \log P_M(z_+) +O\left( \frac1R \right). \label{eq:contour_BC}
\end{equation}

The integral for the component of the path near the origin, which I will call $D$, is similar to that for $B_R$:
\begin{equation}
\frac1{2\pi i} \oint_{D} \frac {P'_M (z) F(z) }{P_M (z) } dz = \int_0^{z_-} \frac {P'_M (z) }{P_M (z) } dz = \log \left( -P_M(z_-) \right) \label{eq:contour_D}
\end{equation}

Plugging Equations~\eqref{eq:contour_BC} and~\eqref{eq:contour_D} into Equation~\eqref{eq:sum_as_int} and taking the limit $R \to \infty$ gives
\begin{equation}
\left| A_M \right| = \frac{ (bt)^{M+1}\sqrt{-P_M(z_+) P_M(z_-)}}{\sqrt{f(1)f(-1)}} = \frac{ (bt)^{M+1}P_M(z_+) }{z_+^{M+1}\sqrt{f(1)f(-1)}}  ,
\end{equation}
where the last equality is a result of $z_+ = 1/z_-$ and the fact that $$P_M(z) = - z^{2M+2} P_M(1/z). $$  Using $f(z_+)=0$ to obtain an expression for $a^2$ in terms of the other variables allows us to obtain $f(1)f(-1) = b^2t^2(z_+^2-1)^2/z_+^2$, and therefore
\begin{equation}
\left| A_M \right| =  \frac{ (bt)^{M}P_M(z_+) }{z_+^{M}(z_+^2-1)} .
\label{eq:new_det}
\end{equation}

This should be compared with the calculation of the same determinant by McCoy and Wu.  They note \cite[pp. 121, 349]{MW} that the expansion of $|A_M|$ in complimentary minors results in a simple recursion relationship, whose solution can ultimately be expressed as
\begin{equation}
|A_M| = 
\begin{bmatrix}
1 & 0
\end{bmatrix}
\begin{bmatrix}
-a^2 + b^2 & at \\
-at & t^2
\end{bmatrix}^M
\begin{bmatrix}
1 \\ 0
\end{bmatrix}. \label{eq:MW_det}
\end{equation}
The characteristic polynomial of the $2 \times 2$ matrix in this expression is
\begin{equation}
\begin{split}
(-a^2+b^2-\lambda)(t^2-\lambda)+a^2t^2 =\lambda^2 - (-a^2+b^2+t^2)\lambda+b^2t^2
\\= -bt \lambda f\left( \frac{\lambda}{bt}\right), 
\end{split}
\end{equation}
so its eigenvalues are 
\begin{equation}
\lambda_\pm = bt z_\pm.
\end{equation}
The corresponding normalized right eigenvectors are
\begin{equation}
\begin{bmatrix}
v_\pm \\ w_\pm
\end{bmatrix}
=
\frac1{\sqrt{(t^2-\lambda_\pm)^2-a^2t^2}}
\begin{bmatrix}
t^2-\lambda_\pm\\
i at
\end{bmatrix}
\end{equation}
so, noting that the matrix involved is Hermitian, Equation~\eqref{eq:MW_det} gives
\begin{equation}
|A_M| = v_+^2 \lambda_+^M + v_-^2 \lambda_-^M 
= \frac{(t^2-btz_+)^2}{(t^2 - bt z_+ )^2-a^2t^2}b^Mt^Mz_+^M +  \frac{(t^2-btz_-)^2}{(t^2 - bt z_- )^2-a^2t^2}b^Mt^Mz_-^M.
\end{equation}
Noting that $z_- = 1/z_+$, we can rewrite this as
\begin{equation}
|A_M| =b^M t^Mz_+^{-M-2} \left[
\frac{(t^2 - bt z_+)^2  z_+^{2M+2}}{a^2t^2+(t^2 - bt z_+ )^2}
+\frac{(t^2z_+ - bt)^2 }{a^2t^2z_+^2+(t^2 z_+ - bt  )^2}
\right].
\end{equation}
Eliminating $a^2$ as before and writing the above expression over a common denominator indeed reproduces Equation~\eqref{eq:new_det}.  Subsitituting this into Equations~\eqref{eq:s_final} and~\eqref{eq:Z_primitive} then gives
\begin{equation}
\begin{split}
Z^2 = 2^{2MN-2}&(\cosh \beta J)^{4MN-2N}  \\ 
\times\prod_{ k = \frac\pi{N}, \frac{3\pi}{N},\dots} & \left| 1+te^{i k }\right|^{2M} [b( k ) t z_+( k )]^M
\\ & \times \left[ \frac{z_+( k )}{z_+( k )^2-1} \right] \left[1 + \frac{b( k ) z_+( k ) -t}{tz_+( k ) -b( k ) } z_+( k )^{-2M-1}  \right] \label{eq:z2_primitive}.
\end{split}
\end{equation}

\section{Universal finite size corrections}\label{sec:uni}
We now turn to the expansion of the pressure $\log Z$ given in Equation~\eqref{eq:kappa_exp}.  To do so we successively examine the factors in the product on the right-hand side of Equation~\eqref{eq:z2_primitive}.

As for the first two factors, noting that $\log b( k )$ is analytic, the Euler-Maclaurin formula gives
\begin{equation}
\begin{split}
 \log  \prod_{ k = \frac\pi{N}, \frac{3\pi}{N},\dots} & \left| 1+te^{i k }\right|^{2} [b( k ) t z_+( k )]
\\ = &\frac{N}{2\pi} \int_0^{2\pi} d k  \left\{ 2 \log \left| 1+te^{i k }\right| +\log [b( k ) t z_+( k )]  \right\}
\\& - \frac{\pi}{N} \frac1{12} \left( \frac{z_+'(0^+) - z_+'(2 \pi^-)}{z_+(0)} \right) + O \left( \frac1{N^2} \right) \label{eq:EM}
\end{split}
\end{equation}
Using the explicit expressions for $z_+$, $a$, and $b$, and setting $t$ to the critical value of $\sqrt2 - 1$, we find (after some tedious algebra) that
\begin{equation}
z_+( k ) = \sqrt{\cos^2( k ) - 4 \cos( k ) + 3} - \cos( k ) + 2, \label{eq:zpc}
\end{equation}
from which it is easy to see that 
 $z_+(0) = 1$ and $z_+'(0^+) = - z_+'(2\pi^-) = 1 $, so that the Equation~\eqref{eq:EM} simplifies to
\begin{equation}
\begin{split}
& \log \prod_{ k = \frac\pi{N}, \frac{3\pi}{N},\dots} \left| 1+te^{i k }\right|^{2} [b( k ) t z_+( k )]
\\& = \frac{N}{2\pi} \int_0^{2\pi} d k  \left\{ 2 \log \left| 1+te^{i k }\right| +\log [b( k ) t z_+( k )]  \right\}
- \frac{\pi}{6} \frac1{N} + O \left( \frac1{N^2} \right).
\end{split}
\end{equation}

The next factor is independent of $M$, and hence of no immediate interest (it gives a contribution to the surface energy of the system, however, as noted by McCoy and Wu \cite[pp.~122-3]{MW}).  

To evaluate the last factor, first note that since $|\log(x)-\log(y)| \le |x-y|$ for $x,y > 1$, 
\begin{equation}
\begin{split}
 \left| \log \left(1 + \frac{b( k ) z_+( k ) -t}{tz_+( k ) -b( k ) } z_+( k )^{-2M-1}  \right)
- \log \left(1 +  z_+( k )^{-2M} \right) \right| 
\\ \le 
\left|
1 - \frac{b( k ) z_+( k ) -t}{tz_+( k ) -b( k ) } z_+( k )^{-1}
\right| z_+( k )^{2M}
\end{split}
\end{equation}
which (since $z_+( k ) < 1$ for $0 <  k  < 2 \pi$) vanishes faster than any power of $M$.  As a result we can replace the last factor with the product over $ k $ of $1+z_+( k )^{-2M}$, and noting that Equation~\eqref{eq:zpc} also gives $z_+''(0) =1$ we can expand $z_+$ and rearrange the terms in order to approximate this product by
\begin{equation}
 \prod_{ k = \frac\pi{N}, \frac{3\pi}{N},\dots} \left[
1+z_+( k )^{-2M}
\right]
=
\left(
 \prod_{ k = \frac\pi{N}, \frac{3\pi}{N},\dots} \left[
1+e^{-2M  k }
\right]
\right)^2 + O \left( \frac{M}{N^2} \right).
\end{equation}
As noted in~\cite{FF}, the product on the right hand side can be expressed in terms of Jacobi theta functions as
\begin{equation}
\prod_{r = 0}^{\infty} \left[ 1 + \exp\left(\frac{-2M(2r-1)\pi}{N}\right) \right]^2
= \frac{\theta_3(e^{-2\pi\frac{M}{N}})}{\theta_0(e^{-2\pi\frac{M}{N}})}.
\end{equation}
where $\theta_j (q)$ is an abbreviation for $\theta_j(0,q)$, and
\begin{equation}
\theta_0(q) = q^{-1/12}\left[\tfrac{1}{2} \theta_2(q) \theta_3(q) \theta_4(q) \right]^{1/3}.
\end{equation}
Subsitituting this and Equation~\eqref{eq:EM} into Equation~\eqref{eq:z2_primitive}, taking the logarithm and isolating the term of relevant order gives 
\begin{equation}
\kappa (\zeta) = \frac16 \log \left(
\frac{ \theta_3^2 \left( e^{-2\pi \zeta} \right) }
{2 \theta_2  \left( e^{-2\pi \zeta} \right) \theta_4  \left( e^{-2\pi \zeta} \right) }
\right). \label{eq:kappa_final}
\end{equation}
where $\zeta = M/N$.

To compare this result with the predictions of conformal field theory, we need to examine the leading-order behavior of $\kappa$ as $\zeta\to \infty$ (i.e.\ $M >> N$) and $\zeta \to 0$ (i.e.\ $N >> M$).  For the former limit, where $q := e^{-2\pi \zeta} \to 0$, we use the series representations of the theta functions to write
\begin{gather}
\theta_2(q) = q^{1/4} + O(q^{9/4}) \\
\theta_3(q) = 1 + O(q) \\
\theta_4(q) = 1+ O(q),
\end{gather}
so that
\begin{equation}
\kappa(\zeta) \to \frac16 \log q^{-1/4} = \frac{\pi}{12} \zeta.
\end{equation}

For the limit $\zeta \to 0$, we use Jacobi's imaginary transformation, which in this case takes the form 
\begin{gather}
\theta_2 (e^{- \pi x})= x^{-1/2} \theta_4 (e^{- \pi / x}) \\  
\theta_3 (e^{- \pi x})= x^{-1/2} \theta_3 (e^{- \pi / x}) \\ 
\theta_4 (e^{- \pi x})= x^{-1/2} \theta_2 (e^{- \pi / x}),
\end{gather}
 to rewrite Equation~\eqref{eq:kappa_final} as
\begin{equation}
\kappa(\zeta) =  \frac16 \log \left(
\frac{ \theta_3^2 \left( e^{-\pi/ 2 \zeta} \right) }
{2 \theta_2  \left( e^{-\pi/ 2 \zeta} \right) \theta_4  \left( e^{-\pi / 2 \zeta} \right) } 
\right) = \frac{\pi}{48} \zeta^{-1} + O(\zeta^{-2}),
\end{equation}
giving the desired term.

\bibliographystyle{alpha}
\bibliography{../Ising,palindromes}

\end{document}